\documentclass[envcountsame]{llncs}
\usepackage{ifpdf}
\ifpdf
\usepackage{pdfsync}
\fi
\usepackage{a4wide}
\usepackage{amsmath}
\usepackage{amssymb}
\usepackage{amsfonts}
\usepackage{mathdots}
\usepackage{mathrsfs}
\usepackage{enumerate}
\usepackage{gastex}
\usepackage{xcolor}
\usepackage{latexsym}
\usepackage{url}
\usepackage{comment}
\usepackage{color}
\usepackage{graphicx}
\usepackage{epstopdf}
\usepackage{algorithm}
\usepackage{algpseudocode}

\usepackage{tikz}
\usetikzlibrary{arrows,positioning, calc, decorations.markings,shapes.multipart,chains}

\algtext*{EndWhile}
\algtext*{EndIf}
\algtext*{EndFunction}
\algtext*{EndFor}

\newcommand{\Can}{\mathsf{C}}
\newcommand{\changeval}{\mathsf{changeval}}

\newcommand{\cut}{\mathsf{cut}}

\newcommand{\deletemin}{\mathsf{deletemin}}
\newcommand{\down}{\mathsf{down}}

\newcommand{\expose}{\mathsf{expose}}

\newcommand{\It}{\mathrm{It}}
\newcommand{\join}{\mathsf{join}}

\newcommand{\lef}{\mathsf{le}}

\newcommand{\link}{\mathsf{link}}

\renewcommand{\max}{\mathsf{max}}
\renewcommand{\min}{\mathsf{min}}

\newcommand{\pd}{\mathsf{Pd}}
\newcommand{\N}{\mathbb N}
\newcommand{\nil}{\mathsf{null}}
\newcommand{\nxt}{\mathsf{next}}

\newcommand{\Path}{\mathsf{Path}}

\newcommand{\psort}{\mathsf{psort}}

\newcommand{\rotateleft}{\mathsf{rotateleft}}

\newcommand{\rig}{\mathsf{ri}}

\newcommand{\sub}{\mathsf{sub}}

\renewcommand{\sup}{\mathsf{sup}}

\newcommand{\Team}{\mathsf{Team}}

\newcommand{\Time}{\mathsf{T}_\mathsf{exp}}

\newcommand{\upp}{\mathsf{up}}
\newcommand{\val}{\mathsf{val}}

\definecolor{lightgray}{gray}{0.65}
\makeatletter \g@addto@macro\@floatboxreset\centering \makeatother

\begin{document}
\title{Dynamic Partial Sorting}

\author{ Jiamou Liu \and Kostya Ross}
\institute{School of Computer and Mathematical Sciences \\ Auckland University of Technology, New Zealand\\
\email{jiamou.liu@aut.ac.nz},\ \  \email{hsv5433@aut.ac.nz}
}

\maketitle

\begin{abstract}
The dynamic partial sorting problem asks for an algorithm that maintains  lists of numbers under the link, cut and change value operations, and queries the sorted sequence of the $k$ least numbers in one of the lists. We first solve the problem  in $O(k\log (n))$  time for queries and $O(\log (n))$ time for updates using the tournament tree data structure, where $n$ is the number of elements in the lists.  We then introduce a layered tournament tree data structure and solve the same problem in $O(\log_\varphi^* (n)  k\log (k))$ time for queries and $O\left(\log (n)\cdot\log^2(\log (n))\right)$  for updates, where $\varphi$ is the golden ratio and $\log_\varphi^*(n)$ is the iterated logarithmic function with base $\varphi$.
\end{abstract}

\section{Introduction} \label{sec:intro}

{\bf The problem setup.} Many practical applications store data in a collection of key-value pairs where the keys are drawn from an ordered domain. In such applications, queries would be made on the order statistics of values within a subset of keys. Consider as an example the loan data of a library. This can be represented by an ordered map whose keys are structured indices indicating the category of the book, and whose values are the number of times the book was borrowed since acquisition. One possible type of queries involves retrieving the most popular (or least popular) books from various categories or subcategories. Facilitating such queries is an inherently {\em dynamic} problem; firstly, the subsets of the ordered map for whose values order statistics are desired can vary, and secondly, ordered maps typically represent
mutable data, which requires values and keys to change.

\smallskip

Existing algorithms and data structures cannot effectively solve this problem. If we
want to represent an mutable ordered map, the standard solution (a
self-balancing binary search tree) cannot efficiently extract order
statistics about its values. On the other hand, existing selection algorithms for data with structured keys rely on the
data being static, which in a dynamic context, would force a re-run of
the algorithm on every change. Neither of these are desirable,
especially when the data being represented by the ordered map is
large. This leads to a need for a solution that can effectively extract order statistics
about values, while being amenable to data mutation.

\smallskip

Abstractly, we may view an ordered map as a list of numbers, where elements are arranged in the list by their keys. The above queries then amount to obtaining order statistics of numbers within intervals of the list. Formally, we propose the {\em dynamic partial sorting problem}, which is stated as follows: \ Maintain a collection of lists $\ell_1,\ell_2,\cdots,\ell_m$ of numbers, while allowing the following {\em partial sorting
operation}:

\begin{itemize}
\item $\psort(\ell_i,k)$: Return the $k$ smallest numbers in $\ell_i$
  if $k$ is at most the size of $\ell_i$, and all elements in $\ell_i$
  otherwise. The output should be in increasing order.
\end{itemize}
We also support the following {\em update} operations:

\begin{itemize}
\item $\changeval(\ell_i,x,x')$: Suppose $x$ is a number in $\ell_i$;
  change $x$ to $x'$.
\item $\link(\ell_i,\ell_j)$: Link the lists $\ell_i$ and $\ell_j$ by
  attaching the tail of $\ell_i$ to the head of $\ell_j$.
\item $\cut(\ell_i,x)$: Suppose $x$ is a number in $\ell_i$; separate
  $\ell_i$ into two lists, such that the first list contains all
  elements from the head of $\ell_i$ to $x$ inclusive, and the second
  list contains all other elements of $\ell_i$.
\end{itemize}
We assume the parameter $x$ in the $\cut(\ell_i,x)$ and $\changeval(\ell_i,x,x')$ operations points directly to the element $x$ in $\ell_i$, and therefore no searching is necessary. In this paper, we are only going to focus on the $\link$, $\cut$, $\changeval$ and $\psort$ operations as defined above. We observe that these operations also permit partial sorting on arbitrary intervals in a list.

\smallskip

Dynamically maintaining a sorted list of numbers is a well-explored topic. Existing solutions include utilizing various self-balancing binary search trees \cite{Sahni05}. These data structures are not suitable for the dynamic partial sorting problem, as here we require elements in the lists to preserve their orders while extracting order statistics from the lists. To the authors' knowledge, there has not been work formally addressing the dynamic partial sorting problem. Here we describe some naive algorithms for solving the problem:

\smallskip

The first naive solution to the dynamic partial sorting problem is to simply put numbers in a linked list. Thus $\link(\ell,\ell')$, $\cut(\ell,x)$ and $\changeval(\ell,x,x')$ are solved in constant time, but to perform $\psort(\ell,k)$, we run an optimal static algorithm such as quick select \cite{Hoare62} and partial quicksort \cite{Martinez04}, which take time $O(n+k\log (k))$.

\smallskip

The second naive solution to the dynamic partial sorting problem is to store the numbers in each list in a priority queue. This allows us to perform $\psort(\ell,k)$ by repeatedly removing and returning the minimum item, and then re-inserting those items afterwards. The running time  of $\psort(\ell,k)$ is $O(k \log (n))$, where $n$ is the number of elements in $\ell$. We can perform $\link(\ell,\ell')$ and $\cut(\ell,x)$ by successively inserting or deleting elements from the priority queues of the lists. Hence each of these operations takes $O(n\log (n))$.

\smallskip

{\bf Related work.} Bordim et al have employed a partial sorting algorithm to solve problems in common-channel communication over single-hop wireless sensor networks \cite{Bordim02}. Additionally, the problem has been generalized to sorting intervals \cite{Martinez10}. The asymptotic time complexity of partial sorting has been thoroughly studied \cite{Kuba06,Hwang00,Floyd75}.

Several data structures for partial sorting have been described. Navarro and Paredes proposed one such structure in \cite{Navarro10}, but it is optimized for use of memory, rather than time, and is both amortized and online. Duch et al presented another structure in \cite{Duch12} for the selection problem which can be used for partial sorting. However the structure is not dynamic, and depends heavily on the length of the input data. 

\smallskip

{\bf Contribution of the paper.} The goal of the paper is to design a solution to the dynamic partial sorting problem where the query and update operations have better time complexity.  We first describe a solution that is based on the {\em tournament tree} data structure.  A tournament tree of a list of numbers is a full balanced binary tree whose leaves are the elements of the list and the value of every internal node is the minimum of the values of its two children. Hence any node in the tournament tree stores the minimum number in the subtree rooted at this node.
Based on this observation, we perform the $\psort(\ell,k)$ operation in time $O(k\log (n))$. We perform $\changeval(\ell,x,x')$ in $O(\log (n))$ time by updating the path from $x$ to the root.
The link and cut operations are handled in a similar way as linking and cutting balanced binary trees, and thus take time $O(\log (n))$.

\smallskip

The tournament tree solution to the partial sorting problem allows efficient query and update operations.  However, the time complexity of the $\psort(\ell,k)$ operation depends both on $k$ and the size $n$ of the list $\ell$. In practical applications where $n$ could be much larger than $k$, it is desirable to make the running time of the query operation independent from $n$.
Therefore we develop another dynamic algorithm that solves the dynamic partial sorting problem with the following properties:

\begin{itemize}
\item We handle $\psort(\ell,k)$ in such a way that the size $n$ of $\ell$ has minimal influence on the time complexity of the operation.
\item The time complexity of the update operations is not much worse than the tournament-tree-based algorithm above. More precisely, the update operations run in $o(\log^2 (n))$.
\end{itemize}

To this end, we introduce a recursive data type called the {\em layered tournament tree} data structure. The main idea is that, instead of using one tournament tree to store the items in a list, we use multiple {\em layers} of tournament trees. The layers extend downwards.
The top layer consists of the tournament tree of the list. This tournament tree is partitioned into {\em teams} where each team can be viewed as a path segment of the tree. Each of these teams is then represented by  a tournament tree in the  layer below, where elements of the team correspond to leaves in the tree.
 The tournament tree of a team is again partitioned into teams which are represented by tournament trees in the subsequent layer. This process continues until the team consists of only one node.
Since we maintain the tournament trees as balanced trees, we can guarantee that a tree in a particular layer has logarithmic size compared to the corresponding tree in the layer above.

\smallskip

We define the partial sort operations for tournament trees on every layer of the data structure. Using an iterative algorithm that recursively calls the partial sort operation in lower layers, we perform the $\psort(\ell,k)$ operation on the original list $\ell$. The time complexity of the operation is $O\left(\log_\varphi^* (n) k\log (k)\right)$ where $n$ is the number of items in $\ell$, $\varphi=\frac{\sqrt{5}+1}{2}$ is the golden ratio and $\log_\varphi^*(n)$ is the iterated logarithmic function with base $\varphi$ (See Section~\ref{sec:psort} for a definition). Since the function $\log_\varphi^*(n)$ is almost constant even for very large values of $n$, the running time of $\psort(\ell,k)$ is almost independent from $n$.  The time complexity of the $\link(\ell,\ell')$, $\cut(\ell,x)$ and $\changeval(\ell,x,x')$ operations is  $O\left(\log n\cdot\log^2(\log n)\right)$.

\smallskip

{\bf Organization.} Section~\ref{sec:tournament} introduces the tournament tree data structure. Section~\ref{sec:TT} describes the solution to the dynamic partial sorting problem using tournament trees. Section~\ref{sec:ltt} introduces the layered tournament tree data structure. Section~\ref{sec:psort} and Section~\ref{sec:update} discuss the algorithms for the $\psort(\ell,k)$ operation and the update operations using layered tournament trees, respectively. Section~\ref{sec:conc} concludes the paper and discusses future work.

\section{Tournament Trees} \label{sec:tournament}
A {\em list} is an ordered tuple of numbers.
We write a list $\ell$ as $a_1,a_2,a_3,\ldots,a_k$ where $k$ and each element $a_i$ is a natural number. Throughout the paper we assume that the elements in a list are pairwise distinct.

\smallskip

{\bf Trees.}
We assume a pointer-based computation model for our {\em tree data structure}. This means that every node in the tree has a reference that points to its parent.
We normally use $T$ for a tree and $V$ for the set of nodes in $T$.
The {\em size} of a tree $T$ is $|V|$. For every node $v\in V$, we use $p(v)$ to denote the parent of $v$ if $v$ is not the root, and set $p(v)=\nil$ otherwise.

We will use binary trees to represent lists of numbers. The fields of any node $v\in V$ in a binary tree consist of a tuple
\[
    (p(v), \lef(v), \rig(v), \val(v))
\]
where $\lef(v), \rig(v)$ are respectively the left child and right child of $v$. The field  $\val(v)$ is a integer value associated with the node $v$.

We use $T(v)$ to denote the subtree rooted at $v$. A {\em path} is a set of nodes $\{u_0,u_1,\ldots,u_m\}$ where $m\in \N$, $u_0$ is a leaf and $u_{i+1}=p(u_{i})$ for $0\leq i<m$. We call $m$ the {\em length} of the path.
The {\em height} $h(T)$ of a tree $T$ is the maximum length of any path in $T$. 
A binary tree $T$ is {\em balanced} if for every node $v\in V$, $|h(T(\lef(v)) - h(T(\rig(v))|\leq 1$.
A binary tree is {\em full} if every internal node has exactly two children, i.e., the $\lef(v)$ and $\rig(v)$ fields are both non-null.

\smallskip

{\bf Tournament trees.} The tournament tree data structure is inspired by the tournament sort algorithm, which uses the idea of a single-elimination tournament in selecting the next element \cite{TAOCP}. Formally, the data structure is defined as follows:

\begin{definition}
A {\em tournament tree} of a list $\ell$ of numbers $a_1,a_2,a_3,\ldots,a_n$ is a balanced full binary tree $T$ that satisfies the following properties:
\begin{enumerate}
\item The tree has exactly $n$ leaves whose values are $a_1,a_2,\ldots,a_n$ respectively.
\item For every internal node $v\in V$, if $\val(\lef(v))=a_i$ and $\val(\rig(v))=a_j$, then $i<j$ and $\val(v) = \min\{a_i,a_j\}$.
\end{enumerate}
\end{definition}

See Figure~\ref{fig:tournament} for an example of a tournament tree. Intuitively, one can view a tournament tree of a list of numbers as a  binary search tree, where the numbers are stored in the leaves. The key of each leaf in the binary search tree is the index of the number it stores in the list, and the value is the number itself. 

\tikzstyle{level 1}=[sibling distance=8cm]
\tikzstyle{level 2}=[sibling distance=4cm]
\tikzstyle{level 3}=[sibling distance=2cm]
\begin{figure}[H]

\scalebox{.6}{
\begin{tikzpicture}[decoration={markings,mark=at position 3cm with {\arrow[black]{stealth}}},path/.style={*->,>=stealth,postaction=decorate},
					list/.style={rectangle split, rectangle split parts=2,draw,rectangle split horizontal,minimum width=1.5cm, on chain}, start chain]

\node[circle,draw,minimum size=20pt,thin]{\Large 2}
	child[very thick]{node[circle,draw,minimum size=20pt,thin]{\Large 2}
		child[thin]{node[circle,draw,minimum size=20pt]{\Large 3}
			child[very thick]{node[circle,draw,minimum size=20pt,thin](base){\Large 3}}
			child{node[circle,draw,minimum size=20pt]{\Large 6}}}
		child[very thick]{node[circle,draw,minimum size=20pt,thin]{\Large 2}
			child[thin]{node[circle,draw,minimum size=20pt]{\Large 9}}
			child[very thick]{node[circle,draw,minimum size=20pt,thin]{\Large 2}}}}
	child{node[circle,draw,minimum size=20pt]{\Large 4}
		child[very thick]{node[circle,draw,minimum size=20pt,thin]{\Large 4}
			child[very thick]{node[circle,draw,minimum size=20pt,thin]{\Large 4}}
			child[thin]{node[circle,draw,minimum size=20pt,thin]{\Large 7}}}
		child{node[circle,draw,minimum size=20pt]{\Large 8}}};

\node [
	list,
	below=of base,
    yshift=0.5cm
] (3) {\Large 3};

\node[left=of 3](label){\huge $\ell:$};

\node [
	list,
	right=of 3
] (12) {\Large 6};

\node [
	list,
	right=of 12
] (16) {\Large 9};

\node [
	list,
	right=of 16,
	xshift=0.1cm
] (2) {\Large 2};

\node [
	list,
	right=of 2,
	xshift=0.1cm
] (4) {\Large 4};

\node [
	list,
	right=of 4,
	xshift=0.2cm
] (7) {\Large 7};

\node [
	list,
	right=of 7,
	xshift=0.5cm
] (13) {\Large 8};

\node [
	draw,
	on chain,
	inner sep=6.5pt,
	right=of 13,
] (end){};

\draw(end.north east) -- (end.south west);
\draw(end.north west) -- (end.south east);

\draw[path] let \p1 = (3.two), \p2 = (3.center) in (\x1,\y2) -- (12);
\draw[path] let \p1 = (12.two), \p2 = (12.center) in (\x1,\y2) -- (16);
\draw[path] let \p1 = (16.two), \p2 = (16.center) in (\x1,\y2) -- (2);
\draw[path] let \p1 = (2.two), \p2 = (2.center) in (\x1,\y2) -- (4);
\draw[path] let \p1 = (4.two), \p2 = (4.center) in (\x1,\y2) -- (7);
\draw[path] let \p1 = (7.two), \p2 = (7.center) in (\x1,\y2) -- (13);
\draw[path] let \p1 = (13.two), \p2 = (13.center) in (\x1,\y2) -- (end);
\end{tikzpicture}}
\caption{A tournament tree of a list $\ell=3,6,9,2,4,7,8$. Edges in principal paths are bolded.}\label{fig:tournament}
\end{figure}
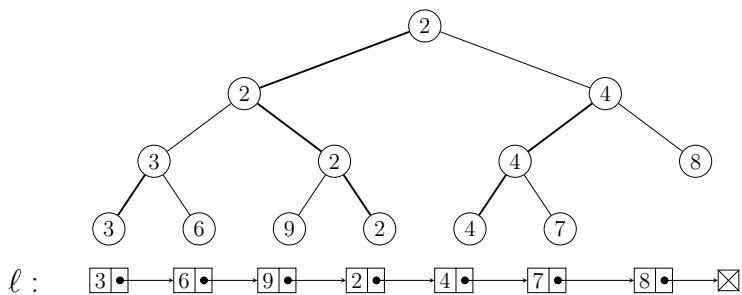

As a tournament tree  is balanced, its height is logarithmic with respect to the number of leaves. More specifically we prove the following lemma.
\begin{lemma}\label{lem:fibonacci}
If $T$ is a tournament tree with $n$ leaves where $n>0$, then the height of $T$ is not more than $\log_\varphi(n)$ where $\varphi$ is the golden ratio.
\end{lemma}
\begin{proof}
It suffices to show that the least number of leaves $f(h)$ in any tournament tree with height $h\geq 0$ is  $\varphi^{h}$, where $\varphi=\frac{\sqrt{5}+1}{2}$ is the golden ratio. The lemma can be easily proved using the following observation. Note that here we use the fact that a tournament tree is balanced and full.
\[f(h)\geq
\begin{cases}
     1 & \text{ if $h=0$,} \\
     2 & \text{ if $h=1$,}\\
     f(h-1)+f(h-2) & \text{ if $h\geq 2$.}
\end{cases}
\]\qed
\end{proof}

\section{Dynamic Partial Sorting With Tournament Trees} \label{sec:TT}

We now describe an algorithm for solving the dynamic partial sorting problem based on tournament trees.  The algorithm assumes that any list $\ell$ of numbers is represented as a tournament tree $T$, whose leaves are the elements of $\ell$. Therefore, we will refer to a list and its tournament tree interchangeably. Furthermore, when we refer to an element $x$ of $\ell$, we also mean the leaf $u$ in $T$ with value $x$ and vice versa. All terms that relate to a tournament tree $T$ carry forward to the corresponding list $\ell$. Hence the {\em nodes}, {\em root}, {\em leaves}, and {\em internal nodes} of $\ell$ refer to the equivalent concepts in $T$.

Let $\ell$ be a list of numbers.
We list the elements of $\ell$ from small to large as $x_1,x_2, \ldots,x_n$.
By definition, the root of $\ell$ has the smallest value. Therefore to find the minimum element $x_1$, we simply return the root. For finding the subsequent $x_i$'s, we make the following definitions.

\begin{definition} Let $T$ be a tournament tree. For any nodes $u,v$ in $T$, we write $u\sim v$ if $\val(u)=\val(v)$.
\end{definition}
As we assume that any list $\ell$ contains pairwise distinct numbers, the equivalence relation $\sim$ partitions the nodes in a tournament tree into disjoint paths.
\begin{definition}
 The {\em principal path} $\Path(u)$ of a node $u$  is the equivalence class $\{v\mid u\sim v\}$. The {\em value} of $\Path(u)$ is $\val(u)$.
\end{definition}
Intuitively we view $\Path(u)$ as a path that {\em originates} from a leaf in $T$ and extends upwards, and every node in $\Path(u)$ ``gains'' its value from this leaf. Hence we single out this leaf and define the following.
\begin{definition}
The {\em origin} of a principal path $P$ is the leaf in $P$.
\end{definition}
Later when referring to ``a principal path'' in the tree $T$, we mean $\Path(u)$ for some node $u$ in $T$.
Note that the second least number in $T$ is the value of a sibling of some node in the principal path of $T$'s root.
In general, for any $1\leq i< n$, let $P_i$ denote the principal path in $T$ with value $x_i$. The number $x_{i+1}$ is $\val(u)$ where $u$ is a sibling of some node in
\[
    P_1\cup P_2\cup \cdots \cup P_{i}.
\]
Hence in computing the $(i+1)$th smallest number in $\ell$ one would need to examine all principal paths whose origins are $x_1,x_2,\ldots,x_i$, and the values of the siblings of nodes on these paths. Formally, we make the following definition.
%
%
\begin{definition} Let $u$ be an internal node in a tournament tree $T$.
The {\em subordinate} $\sub(u)$ of $u$ is a child of $u$ that does not belong to the same principal path as $u$.
\end{definition}

Based on the above observation, to perform $\psort(\ell,k)$, we first output the root of $\ell$ (along with its value), and then apply the following: \ Whenever a node $u$ is returned, we continue to examine the subordinates of all nodes in the principal path of $u$.  This process is continued until we return $\min\{k,n\}$ nodes in $\ell$. During this process we use a priority queue to store the nodes examined so far.
Formally we describe the operation in Algorithm~\ref{alg:tournament-psort}. 

\begin{algorithm}[H]
\caption{$\mathsf{\psort}(\ell,k)$ }\label{alg:tournament-psort}
\begin{algorithmic}[1]
\State $u \gets$ the root of $\ell$
\State Make a new priority queue $Q$
\For{$k$ iterations}
\State Output $\val(u)$
\While{$u \neq \nil$}
\State $y \gets \sub(u)$
\State \Call{$\mathsf{insert}$}{$Q, y$}
\State $u \gets$ the child of $u$ with the same value as $u$, or $\nil$ if no such child exists
\EndWhile
\State $u \gets$ \Call{$\deletemin$}{$Q$}, or $\nil$ if $Q$ is empty
\EndFor
\end{algorithmic}
\end{algorithm}




To perform $\changeval(\ell,x,x')$, we first change the value of the leaf $x$  to $x'$. This can make the values of every ancestor of $x$ incorrect; thus, we walk the path from $x$ to the root of $\ell$, and set the value of every ancestor of $x$ to be the minimum value of its children.  For an exact description, see Algorithm~\ref{alg:tournament-changeval}.

\begin{algorithm}[H]
\caption{$\changeval(\ell,x,x')$}\label{alg:tournament-changeval}
\begin{algorithmic}
\State $\val(x) \gets x'$; $v \gets p(x)$
\While{$v\neq \nil$}
\If{$\val(v) \neq \min\{\val(\lef(v)), \val(\rig(v))\}$}
\State $\val(v) \gets \min\{\val(\lef(v)), \val(\rig(v))\}$
\State $v \gets p(v)$
\EndIf
\EndWhile
\end{algorithmic}
\end{algorithm}

The link and cut operations are handled in a similar way as linking and cutting self-balancing binary search trees as described in \cite{Tarjan85}.

\begin{itemize}
\item {\bf Link.} For the $\link(\ell,\ell')$ operation, we let $T_1$ and $T_2$ denote the tournament trees of $\ell$ and $\ell'$ respectively.
  Without loss of generality, we assume that $h(T_1) > h(T_2)$; the other case is symmetric.
  We would like to join $T_1$ and $T_2$ so that all leaves in $T_1$ are to the left of the leaves in $T_2$ in the resulting tree.
For this operation, we follow right child pointers from the root of $T_1$ until we reach a node $x$ such that $h(T_1(x)) \leq h(T_2)$. We then cut the subtree $T_1(x)$ away from $T_1$, and replace it with a new node $u$; we set $\lef(u)$ to be $x$, $\rig(u)$ to be the root of $T_2$, and $\val(u)$ as the minimum of the values of $u$'s two children. This change can cause the new tree to become unbalanced, and may also require us to modify the values of the nodes on the path from $u$ to the root. To solve these problems, we walk the path from $u$ to the root; at each node $v$ on the path, we must perform two tasks. Firstly, we check whether $T(v)$ is unbalanced; if it is, we perform a left tree rotation on its right child $v'$ and then we set $\val(v')$ to be the minimum of the values of its children. Secondly, we correct $\val(v)$ to be the minimum of the values of its children. We only need to perform a rotation once for any join, as the height of any subtree of $T_1$ has increased by at most 1 as part of this process. Note that the resulting tree is a balanced full binary tree. See Algorithm~\ref{alg:tournament-join}. In this description, we use $\mathsf{rotateleft}(u)$ to refer to a left tree rotation of the node $u$.

\begin{algorithm}
\caption{$\mathsf{\link}(T_1,T_2)$ \ \ (For the $h(T_1)>h(T_2)$ case)}\label{alg:tournament-join}
\begin{algorithmic}[1]
\State $x \gets$ the root of $T_1$, $x' \gets$ the root of $T_2$
\While{$h(T_1(x)) > h(T_2)$}
\State $x \gets \rig(x)$
\EndWhile
\State Create a new node $u$
\State $\rig(p(x)) \gets u$, $\lef(u) \gets x$, $\rig(u) \gets x'$ \Comment{Form a new tree with left subtree $T_1(x)$ and right subtree $T_2$}
\State $\val(u) \gets \min\{\val(u), \val(x')\}$
\State $y \gets u$
\While{$y \neq \nil$}
\State $z \gets \lef(p(u))$
\If{$h(y) > h(z)+1$}
\State \Call{$\mathsf{rotateleft}$}{$y$}
\State $\val(p(z)) \gets \min\{\val(z),\val(\rig(p(z))\}$
\EndIf
\State $\val(y) \gets \min\{\val(\lef(y)), \val(\rig(y))\}$
\State $y \gets p(y)$
\EndWhile
\end{algorithmic}
\end{algorithm}

\item {\bf Cut.} To perform the $\cut(\ell,x)$ operation, we need to split the tournament tree $T$ of $\ell$ at the leaf $u$ where $\val(u)=x$, such that $u$ and all leaves to its left belong to one tournament tree, and all leaves to its right belong to another.  For this operation, we first walk the path from $u$ to the root, deleting every edge on the path and incident to it. We also remove any internal nodes which have no children as part of this process. This breaks the tree into a collection of subtrees, the root of each of which was a child of a node on the path from $u$ to the root. We then link the subtrees containing nodes to the left of $u$ (and $u$ itself)  to form a tournament tree $T_1$, and the subtrees containing the other nodes to form another tournament tree $T_2$.
See Algorithm~\ref{alg:tournament-cut}.

\begin{algorithm}[H]
\caption{$\mathsf{\cut}(T,u)$ }\label{alg:tournament-cut}
\begin{algorithmic}[1]
\State $x \gets p(u)$; $y \gets u$
\State Create two empty tournament trees $T_1, T_2$
\State $T_1 \gets$ $T(y)$
\While{$x \neq \nil$}
\If{$y = \lef(x)$}
\State $T_2 \gets$ \Call{$\link$}{$T_2, T(\rig(x))$}
\Else
\State $T_1 \gets$ \Call{$\link$}{$T(\lef(x)), T_1$}
\EndIf
\State $y \gets x$; $x \gets p(x)$
\EndWhile
\end{algorithmic}
\end{algorithm}
\end{itemize}

\begin{theorem}\label{thm:tournament}
There is an algorithm that solves the dynamic partial sorting problem which performs the $\psort(\ell,k)$ operation in time $O(k\log (n))$, and performs the $\link(\ell,\ell')$, $\cut(\ell,x)$ and $\changeval(\ell,x,x')$ operations in time $O(\log (n))$, where $n$ is the size of the list $\ell$.
\end{theorem}
\begin{proof}
We analyze the time complexity of  the above operations.
\begin{itemize}

\item[(a)] {\bf $\psort(\ell,k)$.} By Lemma~\ref{lem:fibonacci}, every path of the tournament tree is bounded by $\log_\varphi (n)$. This means that when the $\psort(\ell,k)$ operation outputs an element $x$, it inserts at most $\log_\varphi n$ nodes into the priority queue. Hence the priority queue has size bounded by $k\log_\varphi n$. If we use an efficient priority queue implementation, the time complexity of the operation is $O(k\log (n))$.


\item[(b)] {\bf $\changeval(\ell,x,x')$.} By Lemma~\ref{lem:fibonacci} we must modify at most $\lceil \log_\varphi(n) \rceil + 1$ nodes, and each modification consists of an assignment and a two-way comparison, each of which takes constant time. Thus, we have at most $2 (\lceil \log_\varphi(n) \rceil + 1)$ constant-time operations, which makes $\changeval(\ell,x,x')$ an $O(\log (n))$ operation.

\item[(c)] {\bf $\link(\ell,\ell')$.} Let $T_1,T_2$ be the tournament trees of $\ell$ and $\ell'$ respectively. Let $m=|h(T_1)-h(T_2)|$. As discussed above, the $\link(T_1,T_2)$ operation performs at most one rotation and up to $m$ many changes to the values of nodes while walking the path from $u$ to the root. Therefore the $\link(T_1,T_2)$ operation takes time $O(m)$, which is $O(\log (n))$. 

\item[(d)] {\bf $\cut(\ell,x)$.} Let $T$ be the tournament tree of $\ell$ and $u$ be the leaf with value $x$.
Let $P=\{u_0,u_1,u_2,\ldots,u_k\}$ be the path in $T$ from $u_0=u$ to the root of $T$ where $u_{i+1}=p(u_i)$ for all $0\leq i<k$. By Algorithm~\ref{alg:tournament-cut}, the $\cut(T,u)$ operation separates $T$ into  a collection of tournament trees $$\widehat{T}_1, \widehat{T}_2, \ldots, \widehat{T}_k$$ where each $\widehat{T}_i$ is either the left or the right subtree of $u_i$. Since $T$ is balanced, one could easily prove by induction on $i$ that $$h(\widehat{T}_i)\leq 2i-1.$$

The $\cut(T,u)$ operation then iteratively joins the trees $\widehat{T}_1,\ldots,\widehat{T}_k$ to form two trees $T_1$ and $T_2$, where $T_1$ contains all leaves to the left of and including $u$, and $T_2$ contains all leaves to the right of $u$.
We note from (c) that the time required for any $\link$ operation is linear on the height difference  between the two trees being joined.
The total running time of the sequence of $\link$ operations performed is therefore at most
\begin{align*}
    2 \sum_{i\geq 1}^{k-1} \left(h\left(\widehat{T}_{i+1}\right)-h\left(\widehat{T}_i\right)\right) &= 2 \left(h\left(\widehat{T}_k\right)-h\left(\widehat{T}_1\right)\right) \\
                                                                              &\leq 2(2k-1).
\end{align*}
The value of $k$ is at most $h(T)$ which is bounded by $\log_\varphi (n)$. Thus, the total time required for $\cut(T,u)$ is $O(\log(n))$.
\end{itemize}\qed
\end{proof}

\section{Layered Tournament Trees}\label{sec:ltt}

In this section we present an alternative solution to the dynamic partial sorting problem, where the running time of $\psort(\ell,k)$ is (almost) independent from $n$. The algorithm uses a data structure that consists of layers of tournament trees, which we call the {\em layered tournament tree} (LTT) data structure.
Intuitively, the LTT data structure maintains a number of layers that extend downwards, where each layer consists of a number of tournament trees. The tree in the top layer is the tournament tree of $\ell$.  A tree in any lower layer stores a principal path in a tree in the layer above.
Formally, we make the following definitions. Throughout, let $\ell$ be a list of distinct numbers.

\begin{definition}\label{def:team} Let $T$ be the tournament tree of $\ell$. Let $P=\{u_0,u_1,\ldots,u_k\}$ be a principal path in $T$ where $u_0$ is the origin of $P$ and $u_{i+1}=p(u_i)$ for $0\leq i<k$.
We define the {\em team of $P$} as the list of numbers
\[
t=\val(\sub(u_k)),\val(\sub(u_{k-1})),\ldots,\val(\sub(u_{1})).
\]
A {\em team} in the tournament tree $T$ is a team of some principal path $P$ in $\ell$.
\end{definition}
Note that only a principal path with more than one element has a team. We generally use the small case letter $t$ to denote a team.
\begin{definition}\label{def:ltt}
We define a {\em layered tournament tree} ({\em LTT}) of $\ell$ as a set  $\Gamma_\ell$ of tournament trees that satisfies the following:
\begin{itemize}
\item If $\ell$ consists of a single number $x$, then $\Gamma_\ell=\{S\}$ where $S$ consists of a single node whose value is $x$.
\item Otherwise, $\Gamma_\ell$ contains a tournament tree $T$ of $\ell$ as well as an LTT $\Gamma_t$ for each team $t$in $T$. In other words,
\[
\Gamma_\ell= \{T\}\cup \bigcup \left\{\Gamma_{t}\mid t \text{ is a team in } T\right\}.
\]
\end{itemize}
\end{definition}

When the list $\ell$ is clear from the context, we drop the subscript writing $\Gamma_\ell$ simply as $\Gamma$.
We next define layers in a layered tournament tree $\Gamma$ of $\ell$.
\begin{definition}\label{def:layer}  Let $T$ be a tournament tree in  $\Gamma$. We say that
\begin{itemize}
\item  $T$ is {\em in layer $0$} of $\Gamma$ if $T$ is a tournament tree of $\ell$; and
\item  $T$ is {\em in layer $i$} of $\Gamma$, where $i>0$, if $T$ is a tournament tree of a team $t$ in a layer-$(i-1)$ tree in $\Gamma$.
\end{itemize}
We call $\ell$ the {\em layer-$0$ team}, and the team $t$ mentioned above a {\em layer-$i$ team} in $\Gamma$.
If a tree $T$ is in layer $i$ of $\Gamma$, we call it a {\em layer-$i$ tree} in $\Gamma$. The {\em layer number} of $\Gamma$ is the maximum $i\geq 0$ such that a tree is in layer $i$ of $\Gamma$.
\end{definition}
Let $P$ be a principal path in a layer-$i$ tree of $\Gamma$, where $i\geq 0$ and the length of $P$ is at least 1.
By Def.~\ref{def:team} and Def.~\ref{def:layer}, $\Gamma$ contains a tournament tree $T$ of the team of $P$ in layer-$(i+1)$. We call $T$ the {\em team tree} of $P$. The {\em team tree} $\Team(u)$ of any node $u$ is the team tree of the principal path containing $u$. 

Recall that the origin of a principal path $P$ is the leaf in $P$. We introduce the following notions:
\begin{itemize}
\item Suppose $u$ is an internal node in a layer-$i$ tree $T\in \Gamma$. We define $\down(u)$ as the origin $v$ of the principal path in the team tree $\Team(u)$ such that $\val(v)=\val(\sub(u))$.
\item Suppose $u$ is a leaf in a layer-$i$ tree $T\in \Gamma$ where $i>0$. We define $\upp(u)$ as the internal node $v$ in a layer-$(i-1)$ tree such that $\down(v)=u$.
\end{itemize}
This finishes the description of the LTT data structure; see Figure~\ref{fig:LTT} for an example of an LTT.

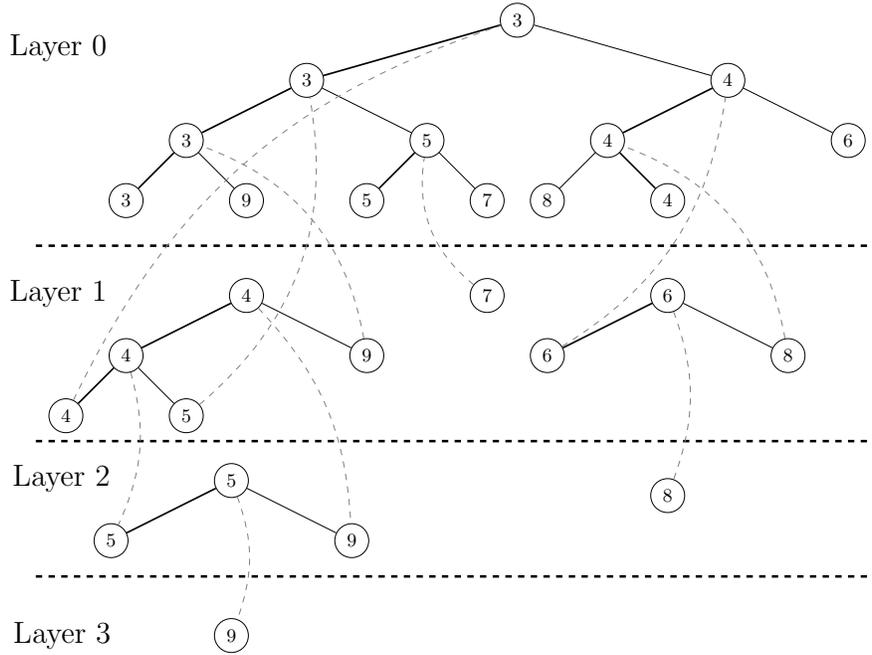
\begin{figure}[!h]
\tikzstyle{level 1}=[sibling distance=7cm]
\tikzstyle{level 2}=[sibling distance=4cm]
\tikzstyle{level 3}=[sibling distance=2cm]
\scalebox{.8}{
\begin{tikzpicture}[tnode/.style={thin,circle,draw},level distance=1cm,emph/.style={edge from parent/.style={black,thick,draw}},norm/.style={edge from parent/.style={black,thin,draw}}]
\node[tnode](root1){3}
	child[emph] {
		node[tnode](first3){3}
		child[emph] {
			node[tnode](3){3}
				child {
					node[tnode](three){3}
				}
				child[norm] {
					node[tnode](mark1){9}
				}
		}
		child[norm] {
			node[tnode](first5){5}
				child[emph] {
					node[tnode]{5}
				}
				child[norm] {
					node[tnode](first7){7}
				}
		}			
	}
	child {
		node[tnode](root3){4}
			child[emph] {
				node[tnode](first4){4}
					child[norm] {
						node[tnode](first8){8}
					}
					child[emph] {
						node[tnode](mark2){4}
					}
			}
			child {
				node[tnode](first16){6}
			}
	};

\tikzstyle{level 1}=[sibling distance=4cm]
\tikzstyle{level 2}=[sibling distance=2cm]

\node[tnode, below=of mark1](root2){4}
	child[emph] {
		node[tnode](four){4}
			child[emph] {
				node[tnode](second4){4}
			}
			child[norm] {
				node[tnode](second5){5}
			}
	}
	child {
		node[tnode](second9){9}
	};
	
\draw[dashed,very thick] (-8,-3.75) to (6,-3.75);

\node[above left=of 3]{\Large Layer 0};

\node[tnode, below=of first7](root5){7};

\node[above left=of four,yshift=-0.5cm,xshift=1cm]{\Large Layer 1};

\node[tnode, below=of mark2](root4){6}
	child[emph] {
		node[tnode](second6){6}
	}
	child[norm] {
		node[tnode](second8){8}
	};
	
\draw[dashed,very thick] (-8,-7) to (6,-7);

\node[tnode, below=of second5, yshift=0.5cm,xshift=0.75cm](third5root){5}
    child[emph]{
        node[tnode](third5leaf){5}}
    child[norm]{
        node[tnode](third9){9}};

\node[above left=of third5leaf,yshift=-0.5cm,xshift=1.3cm]{\Large Layer 2};

\node[tnode, below=of root4,yshift=-1.75cm](third8){8};

\draw[dashed,very thick] (-8,-9.25) to (6,-9.25);

\node[tnode, below=of third5root, yshift=-1cm](fourth9){9};

\node[left=of fourth9,xshift=-0.6cm]{\Large Layer 3};

	
\path[-,gray,dashed]
	(second4) edge[bend left=23] node {} (root1)
	(second5) edge[bend right=30] node {} (first3)
    (second9) edge[bend right=30] node {} (3)
	(root5) edge[bend left=30] node {} (first5)
	(second6) edge[bend right=25] node {} (root3)
	(second8) edge[bend right=30] node {} (first4)
	(four) edge[bend left=20] node {} (third5leaf)
    (third5root) edge[bend left=20] node{} (fourth9)
    (root2) edge[bend left=20] node {} (third9)
    (root4) edge[bend left=20] node {} (third8);

\end{tikzpicture}}\caption{The LTT of the list $\ell=3,9,5,7,8,4,6$.  The $\upp$ and $\down$ nodes are indicated by a dashed grey line.  The layer number is 3.
The team of 3 is a list 4,5,9. The team of 5 is a list with a single element 7. The team of 4 is 6,8. These teams form their own team trees at layer 1.
}\label{fig:LTT}
\end{figure}

\medskip

\noindent {\em Remark.} \
Intuitively the layered tournament tree is similar in concept to a dynamic tree as described by Tarjan and Sleator \cite{Tarjan85}. However by Def.~\ref{def:layer} a dynamic tree has only two layers while a layered tournament tree can have arbitrarily-many.

\medskip

In subsequent sections, we describe the $\psort$, $\link$, $\cut$ and $\changeval$ operations for the LTT data structure.
The factors that determine the time complexity of these operations are \ 1) the height of a layer-$i$ tree in an LTT $\Gamma$ for $i\geq 0$; and \ 2) the layer number in the LTT $\Gamma$.
%

To analyze the height of a layer-$i$ tree in a LTT $\Gamma$ for any $i\geq 0$, we recall the following function.

\begin{definition} Let $b>1$ be a real number.
The {\em iterated logarithm with base $b$} $\log^*_b(n)$ of a number $n>b$ is the smallest $i\geq 0$ such that
$$\underbrace{\log_b\cdots \log_b}_i (n)\leq 1.$$
\end{definition}
It is known that the iterated logarithm function is defined for all $b\leq e^{1/e}$. The function $\log_b^{n}$ is known to be extremely slow-growing; for example, when $b$ is the golden ratio $\varphi$, $\log^*_b(10^6)=6$ and $\log^*_b(10^{10000})=7$.
More precisely, $\log^*_b(n)$ is the inverse of the power tower function with base $b$ defined as
\[
b\uparrow\uparrow n = \underbrace{b^{b^{\iddots^b}}}_{n}
\]
Hence we have the following lemma, which we state without a proof.
\begin{lemma}\label{lem:iterated log}
For any $b\geq e^{1/e}$, for all $i\geq 0$ we have
\[
    \exists n'>0\forall n>n':\ \log^*_b(n) \leq \underbrace{\log_b\cdots \log_b}_i(n).
\]
\end{lemma}

\begin{lemma} \label{lem:number of layers}
For any $i\geq 1$, the size of any layer-$i$ team is at most $\underbrace{\log_\varphi\cdots \log_\varphi}_{i} (n)$, where $n$ is the size of the list $\ell$. Furthermore layer number of the LTT data structure is at most $\log^*_\varphi (n)$.
\end{lemma}

\begin{proof}
By Lemma~\ref{lem:fibonacci} the height of any tournament tree is at most $\log_\varphi (m)$ where $m$ is the number of leaves in the tree. The first statement of the lemma  follows directly from the fact that the number of leaves in a layer-$i$ tree is at most the height of a layer-$(i-1)$ tree in $\Gamma$. The second statement follows directly from the first statement.\qed
\end{proof}
As an example, suppose the list $\ell$ contains a million numbers.  The layer number in the LTT of $\ell$ is at most $\log^*_{\varphi}(10^6)\leq 6$.

\section{The $\psort(\ell,k)$ Operation With LTT} \label{sec:psort}

We now describe the algorithm for solving the dynamic partial sorting problem using the LTT data structure. Similarly to  Section~\ref{sec:TT}, we assume that a list $\ell$ is represented by an LTT $\Gamma$. More specifically, we assume that the elements of $\ell$ are the leaves of the layer-0 tree in $\Gamma$.  In this section we will refer to a list and its LTT interchangeably.
 All terms that relate to a team tree $T$ carry forward to the corresponding list $\ell$. Hence the {\em nodes}, {\em root}, {\em leaves}, and {\em internal nodes} of $\ell$ refer to the equivalent concepts in $T$.

\smallskip

We describe the partial sorting operation $\psort(\ell,k)$ on an LTT $\Gamma$ of the list $\ell$. We use $$x_1<x_2<\ldots<x_n$$ to denote the numbers in $\ell$ in ascending order.
Intuitively the algorithm is similar to the $\psort(\ell,k)$ operation described in Section~\ref{sec:TT}.
The algorithm searches for and outputs each $x_i$ iteratively by exploring the layer-0 tournament tree $T$. The smallest number $x_1$ is the value of the root of $T$. If $k=1$ or $\ell$ contains only one element, then the algorithm  stops after outputting $x_1$. Otherwise, to find the second-smallest number $x_2$ in $\ell$, let $P$ be the principal path of the root of $T$. The number $x_2$ is the least number in the team of $P$. Unlike Algorithm~\ref{alg:tournament-psort}, where we check through the subordinates of all nodes in $P$, here we recursively apply the partial sort operation on the tournament tree of the layer-1 team of $P$. In this way, the search continues in a lower layer.

To formally describe the $\psort(\ell,k)$ operation, we use an {\em iterator}, which is defined as follows.

\begin{definition}\label{def:iterator} Let $\ell$ be a list of numbers.
An {\em iterator of $\ell$} is a data structure $\It(\ell)$ that supports an operation $\nxt(\ell)$ with the following property: \ \ In the $i$th call to $\nxt(\ell)$, the operation outputs $x_i$ if $i\leq n$; and outputs $\nil$ otherwise.
\end{definition}


An iterator $\It(\ell)$ maintains a priority queue $Q$, which is going to contain nodes in $T$.
The $\psort(\ell,k)$ operation amounts to creating an iterator $\It(\ell)$ and calling $\nxt(\ell)$ $k$ times to obtain the list $x_1,x_2,\ldots,x_k$.
We use $u_i$ to denote the leaf with value $x_i$ in the layer-0 tree of $\ell$ for $1\leq i\leq n$.
For convenience, we consider the output of $\nxt(\ell)$ to be the leaf $u_i$, rather than its value $x_i$.

%

To create an iterator for $T$, the algorithm simply creates an empty priority queue $Q$.
We describe the $\nxt(\ell)$ operation by induction on the number of elements in $\ell$.
When the operation $\nxt(\ell)$ is called the first time, we return the origin of $\Path(r)$, where $r$ is the root of $\ell$.
In subsequent calls to $\nxt(\ell)$, if $\ell$ contains only one element, then the algorithm returns $\nil$. Suppose $\ell$ contains more than one element, and assume that we have defined iterators of lists with fewer elements than $\ell$.

Suppose $i\geq 1$ and we have made $i$ calls to $\nxt(\ell)$ which outputs the nodes
$$ u_1,u_2,\ldots,u_i $$.
Algorithm~\ref{alg:next} implements the $\nxt(\ell)$ operation for the $(i+1)$th call.
%
\begin{algorithm}[H]
	\caption{$\nxt(\ell)$ \qquad (The $(i+1)$th call)}\label{alg:next}
	\begin{algorithmic}[1]
    \If{$\Team(u_i)$ is not empty}
        \State Create an iterator $\It(\Team(u_{i}))$
        \State $a\gets \nxt(\Team(u_{i}))$
        \State Insert $\upp(a)$ to $Q$ with value $\val(\sub(a))$
    \EndIf
    \If{$Q$ is not empty}
        \State $x\gets \deletemin(Q)$
        \State $u_{i+1}\gets$ the origin of $\Team(\sub(x))$
        \State $b\gets \nxt(\Team(x))$
        \If{$\upp(b) \neq \nil$}
        \State Insert $\upp(b)$ to $Q$ with value $\val(\sub(b))$
        \EndIf
        \State Output $u_{i+1}$
    \Else
        \State Output $\nil$
    \EndIf
	\end{algorithmic}
\end{algorithm}

To show the correctness of the algorithm  above, we make the following definition:
\begin{definition}
Let $v$ be a node in a tournament tree $T$. The {\em superordinate} of $v$ is a node $\sup(v)$ in $T$ whose subordinate belongs to the principal path $\Path(v)$.
The {\em superordinate set} of a set $U$ of nodes is
\[
    \sup(U) = \{\sup(v)\mid v\in U\}.
\]
\end{definition}
For the next definition, we take a set $U$ of nodes in $T$.
\begin{definition}
A node $v$ is an {\em $U$-candidate} if there is some $u\in U$ such that $v\in \Path(u)$ and for any $w\in \Path(u)$, $\val(\sub(w))<\val(\sub(v))$ if and only if $w\in \sup(U)$.
We denote the set of $U$-candidates as $\Can(U)$.
\end{definition}
\begin{lemma}\label{lem:gamma1}
For every $1\leq i<n$, $\sup(u_{i+1})\in \Can(\{u_1,\ldots,u_{i}\})$.
\end{lemma}

\begin{proof}
We prove this lemma by induction on $i$. By definition of the tournament tree $T$, $u_2$ is the subordinate of a node  $v\in\Path(u_1)$. Furthermore, $\val(u_2)$  is the smallest number in the team of $\val(u_1)$. Hence $\sup(u_2)\in \Can(\{u_1\})$.

Suppose the statement holds for $i\geq 1$. Let $x$ be the superordinate of the node $u_{i+1}$.
Our goal is to show that $x\in \Can(\{u_1,\ldots,u_i\})$. For any node $v\in \Path(x)$, we have $\val(v)<\val(u_{i+1})$ as otherwise $v$ would not be in the same principal path as $x$. Hence the head of the principal path $\Path(x)$ is $u_j$ for some $1\leq j\leq i$.

Let $w$ be a node in $\Path(x)$. Suppose $\val(\sub(w))<\val(\sub(x))$. Since $\val(\sub(x))=\val(u_{i+1})$,  the team of $\Path(\sub(w))$ would contain a number that has smaller value than $u_{i+1}$. Therefore $w$ must be $\sup(u_j)$ for some $1\leq j\leq i$. This means that $w\in \sup(\{u_1,\ldots,u_i\})$. Conversely, suppose $w\in \sup(\{u_1,\ldots,u_i\})$. Then by choice of $u_{i+1}$ we have $\val(\sub(w))<\val(u_{i+1})=\val(\sub(x))$. Thus $x$ is in $\Can(\{u_1,\ldots,u_i\})$. \qed
\end{proof}
The next lemma implies the correctness of Alg.~\ref{alg:next}.

\begin{lemma}\label{lem:next0-correct}
For any $i\geq 1$, the $i$th call to $\nxt(\ell)$ returns the node $u_i$ if $i\leq n$, and $\nil$ otherwise.
\end{lemma}

\begin{proof}
We prove the lemma by induction on the number of times $\nxt(\ell)$ is called.
It is clear that in the first call to $\nxt(\ell)$, the algorithm returns the node $u_1$ which is the origin of the principal path that contains the root of $\ell$.
Consider the second call to $\nxt(\ell)$. If $\ell$ contains only one node $u_1$, then $\Team(u_1)$ does not exist and the priority queue $Q$ is empty at line 5. If $\ell$ contains more than one element, then $\Team(u_1)$ is defined.  At line 5, $Q$ will store the element $x=\upp(a)$, where $a=\nxt(\Team(u_1))$ is the node with the smallest value in $\Team(u_1)$. By definition $\Can(\{v_1\})=\{x\}$.

\smallskip

For the inductive step, suppose we are running $\nxt(\ell)$ the $(i+1)$th time, where $i\geq 1$.  We assume the following inductive assumption: \ When the algorithm reaches line 5,
\begin{enumerate}
\item[(I1)] if $\ell$ contains no more than $i$ elements, then the priority queue $Q$ is empty;
\item[(I2)] if $\ell$ contains at least $i+1$ elements, then the priority queue $Q$ contains exactly those nodes in $\Can(\{u_1,\ldots,u_i\})$.
\end{enumerate}
If $\ell$ contains no more than $i$ elements, then by (I1) the algorithm returns $\nil$ and $Q$ remains empty.
Now suppose $\ell$ contains at least $i+1$ elements. By (I2), when the algorithm reaches line 5, the priority queue $Q$ contains exactly those nodes in $\Can(\{u-1,\ldots,u_i\})$. Let $x$ be the least element in $Q$. By Lemma~\ref{lem:gamma1}, $x$ is the superordinate  $\sup(u_{i+1})$ of $u_{i+1}$. Thus the algorithm would locate and return the node $u_{i+1}$.

\smallskip

We then need to verify that the $\nxt(\ell)$ operation preserves the inductive invariants (I1) and (I2). It is clear that $(I1)$ holds at line 5 of the $(i+2)$th call to $\nxt(\ell)$.

To verify (I2), let $S$ and $S'$ denote the sets of nodes stored in the priority queue $Q$ at line 5 in the $(i+1)$th and the $(i+2)$th call to $\nxt(\ell)$, respectively.
Let $b$ be the leaf that has the next smallest value in $\Team(x)$ after $x$. After we finish the $(i+1)$th call to $\nxt(\ell)$, $Q$ would store the set $S\setminus\{x\}\cup \{\upp(b)\}$. In the $(i+2)$th  call to $\nxt(\ell)$, before reaching Line 5, the algorithm would add the node $\upp(a)$ to $Q$  where $a$ has the least value in $\Team(u_{i+1})$. Therefore we have $$S'=S\setminus \{x\}\cup \{\upp(a),\upp(b)\}= \Can(\{u_1,\ldots,u_i,u_{i+1}\}).$$ Hence (I2) is preserved.
\qed
\end{proof}
As described above, the $\psort(\ell,k)$ operation amounts to creating an iterator of $\ell$ and calling the $\nxt(\ell)$ operation $k$ times. By Lemma~\ref{lem:next0-correct}, the operation outputs the desired numbers $x_1,x_2,\ldots,x_k$ in increasing order.

\smallskip

{\bf Time complexity.} \ We now analyze the time complexity of the $\psort(\ell,k)$ operation.
Suppose $t$ is a layer-$i$ team in $\Gamma$.
 Any call to the $\nxt(t)$ operation may in turn trigger a sequence of calls to the $\nxt(t')$ operations on teams in lower layers. The algorithm maintains a priority queue for every team for which an iterator is created.

Each call to $\nxt(t)$ performs  a fixed number of priority queue operations (such as insert and $\deletemin$), at most two calls to the $\nxt(t')$ operation on some layer-$(i+1)$ team $t'$, and a fixed number of other elementary operations. Among these operations, the first call to $\nxt(t')$ occurs immediately after the $(i+1)$-iterator of $t'$ is created. This call to $\nxt(t')$ simply involves a pointer lookup and thus takes constant time. Furthermore, by Lemma~\ref{lem:fibonacci}, the number of leaves of the team tree of $t'$ is at most $\log_{\varphi}(m)$ where $m$ is the number of elements in $t$.

Suppose we perform $k$ calls to $\nxt(t)$ where $k\geq 1$.
Note that for any team $t'$ in layer $j>i$, the algorithm would make at most $k-1$ calls to $\nxt(t')$.
With every call to $\nxt(t')$, the number of elements stored in the priority queue increases by at most 2. Thus the number of elements stored in any priority queue  is at most than $2k$. Therefore, using a heap implementation of priority queues, the time for inserting an element to or deleting the minimum element from the priority queue takes $O(\log (k))$.

 Summing up the above costs over all $k$ calls, the operations perform $O(k)$ number of priority queue operations, $k-1$ calls to $\nxt$ on trees in a layer down, and other operations that take a total of $O(k)$ time. We use $\mu(k,m)$ to denote the time taken by $k$ calls to $\nxt(t)$ where the team tree of $t$ has $m$ leaves.
Assuming an efficient priority queue implementation, there is a constant $d>0$ such that.
\begin{equation}\label{eqn:T_i}
\mu(k,m)\leq \begin{cases}
    dk\log k+\mu(k-1, \log_\varphi (m)) & \text{if $m>1$;} \\
    d & \text{otherwise.}\end{cases}
\end{equation}

\begin{lemma}\label{lem:psort}  The $\psort(\ell,k)$ operation runs in time $O( \log^*_\varphi (n)k\log (k))$ where $n$ is the size of the list $\ell$.
\end{lemma}

\begin{proof} The $\psort(\ell,k)$ operation makes $k$ calls to the $\nxt(\ell)$ operation. Therefore the running time of $\psort(\ell,k)$ is $\mu(k,n)$ where $n$ is the size of $\ell$. By (\ref{eqn:T_i}) we get
\[
\mu(k,n) \leq dk\log (k) + d(k-1)\log (k) + d(k-2)\log (k) +\cdots +d(k-s+1)\log (k) + d,
\]
where $n$ is the number of elements in $\ell$ and $s$ is layer number of $\Gamma$.
By Lemma~\ref{lem:number of layers}, $s\leq \log^*_\varphi(n)$.
Therefore the total time taken by $\psort(\ell,k)$ is $O(\log^*_\varphi(n)k\log (k))$.\qed
\end{proof}

\section{The Update Operations With LTT} \label{sec:update}
We describe the update operations assuming that all lists are represented by the LTT data structure.
Unless stated otherwise, all occurrences of $\link,\cut$ and $\changeval$ refer to the update operations defined in this section, but not to the operations with the same names in Section~\ref{sec:TT}. As explained in Section~\ref{sec:psort}, the arguments of the $\link$, $\cut$ and $\changeval$ operations consist of LTTs (representing lists) and leaves in the layer-0 tree of the corresponding LTTs (representing elements in the lists).


\smallskip

In the following we define the $\link(\ell,\ell')$, $\cut(\ell,x)$ and $\changeval(\ell,x,x')$ operations by induction on the maximum layer number in the argument LTTs $\ell,\ell'$.
If an LTT consists of only one layer, it contains only one node. Therefore the $\cut$ and $\changeval$ operations performed on such an LTT are trivial. To perform the $\link(\ell,\ell')$ operation where both $\ell,\ell'$ consist of one layer, we create a new node $v$ and set $\lef(v)$ and $\rig(v)$ as $\ell$ and $\ell'$ respectively in the layer-0 tree, and create a layer-1 tree with a single node whose value is the larger of the values of the nodes in $\ell$ and $\ell'$.
In subsequent sections we define the $\changeval(\ell,x,x')$, $\link(\ell,\ell')$ and $\cut(\ell,x)$ operations where $\ell$ and $\ell'$ have more than one layer. The inductive hypothesis assumes correct implementation of $\link$ and $\cut$ on LTTs with fewer layers.


\subsection{The $\expose(\ell,u)$ and $\changeval(\ell,u,x')$ Operation}\label{subsec:expose}

The $\changeval(\ell,x,x')$ operation assumes that $x$ is a leaf in the layer-0 tree of the LTT representing $\ell$ and changes its value to $x'$.
Note that after changing the value of $x$ to $x'$, the LTT structure may be broken. Thus we should apply other procedures to preserve the LTT. This is achieved using an $\expose(\ell,u)$ operation where $u=p(x)$.

Intuitively, the $\expose(\ell,u)$ operation is a ``fix up'' operation that maintains the LTT structure on the path from $u$ to the root of the tree, once a change has occurred on a child.
It walks the path from $u$ to the root, and performs the following procedures in each step: \ It first separates $u$ from its principal path from below, so that both its left child $\lef(u)$ and right child $\rig(u)$ are detached from the principal path of $u$. It then  links the smaller of $\lef(u)$ and $\rig(u)$ with the principal path of $u$ and sets $\val(u)$ as the smaller value of its children. Finally, it repeats the same process to set $p(u)$ as the new $u$.

To separate and link the principal paths mentioned above, we use the $\cut$ and $\link$ operations on the team trees of the corresponding principal paths. Note that in the above operation, we may change the subordinate of $u$. This requires us to change the value of $\down(u)$ in the team tree $\Team(u)$, which can be performed by calling $\changeval(\Team(u), \down(u), \max\{\lef(u),\rig(u)\})$ recursively. Note that the team trees used as arguments of the $\cut,\link$ operations and the recursive call to $\changeval$ have strictly fewer layers than $\ell$. Thus, by the inductive hypothesis, these operations have been defined. For an exact description, see Algorithm~\ref{alg:expose}.

\begin{algorithm}[H]
\caption{$\expose(\ell,u)$ }\label{alg:expose}
\begin{algorithmic}[1]
\State $x\gets u$
\While{$x\neq \nil$ }
    \State $(z,z') \gets (\lef(x),\rig(x))$ if $\val(\lef(x)) < \val(\rig(x))$; otherwise $(z,z')\gets (\rig(x),\lef(x))$ \label{line:min}
    \State $\val(x)\gets \val(z)$ \label{line:change}
    \State $T_1,T_2 \gets$ \Call{$\cut$}{$\Team(x),\down(x)$} \label{line:splt}
    \State $T_1 \gets$ \Call{$\link$}{$T_1, \Team(z)$}\label{line:join}
	\State \Call{$\changeval$}{$T_1, \down(x),\val(z')$} \Comment{Change the value of $\down(x)$ in the layer below}\label{line:correct}
	\State $x\gets p(x)$
\EndWhile
\end{algorithmic}
\end{algorithm}

We now analyze the correctness of the $\expose(\ell,u)$ operation. More specifically, let $v$ be an internal node in the LTT data structure. We use the following invariants:
\begin{enumerate}
\item[(J1)] $\val(v) = \min\{\val(\lef(v)),\val(\rig(v))\}$
\item[(J2)] $\val(\down(v)) = \val(\sub(v))$
\item[(J3)] If $v$ has a child $v'$ that is an internal node and $\val(v)=\val(v')$, then $\down(v),\down(v')$ belong to the same team tree $\Team(v)$ and $\down(v)$ is to the left of $\down(v')$ in $\Team(v)$.
\end{enumerate}
Intuitively, the three invariants state that the LTT structure is maintained.  Indeed,  (J1) states that the value of $v$ is assigned according to the tournament tree property, (J2) states that $\down(v)$ has the correct value, and  (J3) states that the team tree of $\down(v)$ is correctly maintained.
%
\begin{definition}
Let $v$ be a node in the LTT data structure of $\ell$. The {\em parent-down closure} of $v$ is the minimal set $\pd(v)$ of nodes in the LTT that contains $v$ and for any node $w\in \pd(v)$,
\begin{enumerate}
\item $p(w)\in \pd(v)$ if $w$ is not the root of a tree; and
\item $\down(w)\in \pd(v)$ if $w$ is not a leaf in a tree.
\end{enumerate}
\end{definition}
Note that the $\expose(\ell,u)$ operation may only update the values, as well as split and join team trees, for nodes in the set $\pd(u)$. Hence intuitively, $\pd(u)$ denotes the ``region of operation'' in the LTT $\ell$ of $\expose(\ell,u)$. For the next lemma, recall that we assume by the inductive hypothesis that a correct implementation of $\link$ and $\cut$ can be called on LTTs with fewer layers than $\ell$.
\begin{lemma} \label{lem:changeval}
After running $\expose(\ell,u)$, (J1)--(J3) hold for every node $v\in \pd(u)$.
\end{lemma}

\begin{proof}
The proof proceeds by induction on the number of layers in $\ell$. The statement is clear for $\ell$ with a single layer (which consists of only one node). Now suppose $\ell$ contains $m$ layers where $m>1$. Take a node $v\in \pd(u)$ that is in layer-0 of the LTT $\ell$.
Then $v$ is set as $x$ by some iteration of the $\mathsf{while}$-loop. During this iteration, (J1) holds after running Line~\ref{line:change}, (J2) holds after running Line~\ref{line:correct} and (J3) holds after running Line~\ref{line:join} for the node $v$.

Suppose that (J1)--(J3) hold for all nodes in $\pd(u)$ on some layer-$i$ and $v\in \pd(u)$ is an internal node in a layer-$(i+1)$ tree of the LTT data structure.
Then by definition of $\pd(u)$, there is some leaf $w$ in the subtree rooted at $v$ such that $w=\down(w')$ for some $w'\in \pd(u)$. Let $w$ be the rightmost  leaf with this property. The algorithm must have made a call $\changeval(T_1, w,\val(z'))$ during its execution. In this call to $\changeval$, the $\mathsf{while}$-loop visits $v$ and make (J1)--(J3) hold for $v$ using Line~\ref{line:change}, Line~\ref{line:correct} and Line~\ref{line:join} respectively. \qed
\end{proof}
%

\subsection{The $\link(\ell,\ell')$ and $\cut(\ell,x)$ Operations} \label{subsec:multiupdate}

The $\link(\ell,\ell')$ operation is performed similarly to linking two balanced binary search trees. The operation compares the layer-0 trees of $\ell$ and $\ell'$ and links the tree with a smaller height as a subtree of the other.

Before we describe the $\link(\ell,\ell')$ operation, we describe the tree rotation operation for LTTs, which is an important subroutine.  Here, we describe the left rotation $\rotateleft(\ell, u)$, where $u$ is a right child in an LTT $\ell$; the right rotation operation is symmetric. To perform $\rotateleft(\ell, u)$, we first separate both $u$ and its parent $p(u)$ from the rest of their principal paths from above and below. We then perform the left rotation on $u$ as if for a normal binary tree. Lastly, we restore the principal paths of $p(u)$ by calling the $\expose(\ell,p(u))$ operation. This will fix the principal paths we separated in this operation and preserve the structure of the LTT. See Algorithm~\ref{alg:multirotateleft}.

\begin{algorithm}[H]
\caption{$\rotateleft(\ell,u)$}\label{alg:multirotateleft}
\begin{algorithmic}[1]
\State $y \gets p(u)$;
\If{$y$ is not the root}
\State \Call{$\cut$}{$\Team(p(y)), \down(p(y))$} \Comment{Separate $y$  from above}
\EndIf
\State \Call{$\cut$}{$\Team(y),\down(y)$} \Comment{Separate $y$ from below}
\State \Call{$\cut$}{$\Team(u),\down(u)$} \Comment{Separate $u$ from below}
\State $\rig(y)\gets \lef(u)$; $\lef(u)\gets y$ \Comment{Perform the left rotation on $u$}
\State \Call{$\expose$}{$\ell,y$}
\end{algorithmic}
\end{algorithm}
The following lemma is implied from Lemma~\ref{lem:changeval} and the proof is straightforward.

\begin{lemma} \label{lem:multilink}Let $y$ be the parent of $u$.
After running $\rotateleft(\ell,u)$, (J1)--(J3) hold for every node $v\in \pd(y)$.
\end{lemma}

We now describe the $\link(\ell,\ell')$ operation. For simplicity in this section we only describe the case when the layer-0 tree of $\ell$ has a greater or equal height than the layer-0 tree of $\ell'$; the other case is symmetric. We first find a node $u$ on the rightmost path in the layer-0 tree of $\ell$ such that $T(u)$ has the same height as the layer-0 tree $T'$ of $\ell'$. We then create a new node $v$, making it a child of $p(u)$ if $u$ is not the root, and set $T(u)$ as $v$'s left subtree and $T'$ as $v$'s right subtree.  We then fix the principal paths by calling $\expose$ on $v$. This operation may leave the resulting layer-0 tree unbalanced. Hence we walk the path from $v$ to the root and find a node $y$ on this path such that the subtree $T(p(y))$ is unbalanced, and we call $\rotateleft$ on $y$. See Algorithm~\ref{alg:multilink}.
This finishes the description of the $\link(\ell,\ell')$ operations. Note that inside this operation,  all recursive subroutine calls to $\link$ and $\cut$ are made on argument LTTs with fewer layers than $\ell$, and are thus defined by the inductive hypothesis.




\begin{algorithm}[H]
\caption{$\link(\ell,\ell')$ }\label{alg:multilink}
\begin{algorithmic}[1]
\State $T,T' \gets$ the layer-$0$ tournament trees of $\ell,\ell'$ respectively
\State $r_1,r_2 \gets$ the roots of $T,T'$ respectively\
\State Follow $\rig$ pointers from $r_1$ to find $u$ such that $T(u)$ and $T'$ have the same height
\State Create a new node $v$ and the corresponding node $\down(v)$ in the layer below
\State $p(v)\gets p(u)$
\State $\lef(v)\gets u$; $\rig(v)\gets r_2$
\State \Call{$\expose$}{$\ell, v$} \label{line:link-expose}
\State Following $p$ pointers from $v$ until we reach $y$ such that $T(p(y))$ is unbalanced
\State If such $y$ exists, \Call{$\rotateleft$}{$\ell,y$} \label{line:link-rotate}
\end{algorithmic}
\end{algorithm}

We perform the $\cut(\ell,u)$ operation in a similar way as Alg.~\ref{alg:tournament-cut} in Section~\ref{sec:TT}.  The operation first calls $\changeval$ on $u$ to assign it a value smaller than all numbers in $\ell$ (we call it $-\infty$ for convenience). In this way, all nodes on the path from $u$ to the root form a principal path. The operation then walks the path from $u$ to the root,  joining all subtrees to its left into a new tree and all  subtrees to its right into another new tree. Finally it restores the value of $u$ and joins $u$ to the first new tree. We perform all the joining of trees using the $\link$ operation; see Alg.~\ref{alg:multicut}.


\begin{algorithm}[H]
\caption{$\mathsf{\cut}(\ell,u)$ }\label{alg:multicut}
\begin{algorithmic}[1]
\State $a\gets \val(u)$; \Call{$\changeval$}{$\ell,u,-\infty$}
\State $x \gets p(u)$; $y \gets u$
\State Create two empty tournament trees $T_1, T_2$
\While{$x \neq \nil$}
\If{$y = \lef(x)$}
\State $T_2 \gets$ \Call{$\link$}{$T_2, T(\rig(x))$}
\Else
\State $T_1 \gets$ \Call{$\link$}{$T(\lef(x)), T_1$}
\EndIf
\State $y \gets x$; $x \gets p(x)$
\EndWhile
\State $\val(u)\gets a$; \Call{$\link$}{$T_1,u$} \Comment{Link $T_1$ with the restored $u$}
\end{algorithmic}
\end{algorithm}

\subsection{Time Complexity of the Update Operations}

We now analyze the time complexity of the update operations.  For any list $\ell$ with $n$ elements, we define $s_i(n)$ as the maximum number of elements of a layer-$i$ team in the LTT of $\ell$.
It is clear that $s_0(n)=n$. By Lemma~\ref{lem:number of layers}, for all $n>0$ we have
\begin{align}
 &s_{\log^*_\varphi(n)}(n)=1, \text{ and }  \notag\\
 &\forall i\geq 0:\ s_{i+1}(n)\leq \log_\varphi(s_i(n)) \label{eqn:slog}
\end{align}
For convenience, we set $s_i(n)=1$ for all $i>\log^*_\varphi(n)$.
We will express the complexity of the update operations using the variables $s_i(n)$. 

\begin{lemma}\label{lem:log product}
For any $i\geq 0$, there is a constant $n_0>0$ such that for all $n>n_0$ we have
\[
    \prod_{j\geq i+1} s_j(n) \leq s_i(n)
\]
\end{lemma}
\begin{proof}
As $s_j(n)=1$ for all $n>0$ and $j\geq \log_\varphi^*(n)$, the statement is clear for $i\geq \log_\varphi^*(n)-1$.
The proof proceeds by induction on $i$. Fix $0<i<\log_\varphi^*(n)$ and suppose there is $n_0$ such that the statement holds for all $n>n_0$.
Then for all $n\geq n_0$ we have
\begin{align*}
\prod_{j\geq i} s_j(n) &= s_i(n)\cdot \prod_{j\geq i} s_j(n) &\\
                   &\leq  s^2_i(n) & \text{(by the ind. hyp.)}\\
                   &\leq \log^2_\varphi(s_{i-1}(n)) & \text{(by (\ref{eqn:slog}))}\\
\end{align*}
Take $n'$ such that
\[
   \log^2_\varphi(s_{i-1}(n')) \leq s_{i-1}(n').
\]
Then for all $n\geq \max\{n',n_0\}$
\[
    \prod_{j\geq i} s_j(n) \leq \log^2_\varphi(s_{i-1}(n)) \leq s_i(n).
\]
\qed\end{proof}
Recall that the $\expose(\ell,u)$ operation performs a number of iterations. We analyze the running time of each iteration separately.
Without loss of generality, we assume in the next lemma that the list $\ell$ contains no fewer elements than $\ell'$.
\begin{lemma}\label{lem:update} Let $n$ be the number of elements in the list $\ell$. The following hold for the update operations:
\begin{enumerate}
\item[(a)] Each iteration of $\expose(\ell,u)$ runs in time $O\left(s_2^2(n)\right)$.
\item[(b)] The $\expose(\ell,u)$ and $\changeval(\ell,u,x')$ operations run in time $O\left(s_1(n)s_2^2(n)\right)$.
\item[(c)] The $\join(\ell,\ell')$ operation runs in time $O\left(d(\ell,\ell')\cdot s_2^2(n)\right)$ where $d(\ell,\ell')$ is the height difference between the layer-0 trees of $\ell$ and $\ell'$.
\item[(d)] The $\cut(\ell,u)$ operation runs in time $O\left(s_1(n)s_2^2(n)\right)$.
\end{enumerate}
\end{lemma}

\begin{proof}
We prove the lemma by induction on the layer number of $\ell$. The statements are clear if $\ell$ consists of a single layer. For the case when $\ell$ has more than one layer, we prove each statement as follows:

\begin{enumerate}
\item[(a)]  We use $\Time(n,0)$ to denote the maximal running time of each iteration of $\expose(\ell,u)$. It is clear that the number of iterations is bounded by the length of the path from $u$ to the root, which is at most $s_1(n)$. Hence the total running time of $\expose(\ell,u)$ is $s_1(n)\Time(n,0)$.

    Note also that each iteration of $\expose(\ell,u)$ may make a recursive call to $\expose$ on a team in the layer below, and this recursive call may trigger further recursive calls to $\expose$ on lower layers of the LTT. Thus for $0\leq i\leq \log^*_\varphi(n)$ and any layer-$i$ team $t$,
     we define $\Time(n,i)$ as the maximal running time of an iteration in a recursive call $\expose(t,v)$ that is made within $\expose(\ell,u)$. Since the recursive call $\expose(t,v)$ consists of at most $s_{i+1}(n)$ iterations, the total running time of $\expose(t,v)$ is at most $s_{i+1}(n)\Time(n,i)$.

    To prove (a), we prove by induction on $i$ that $\Time(n,i)$ is $O(s_{i+2}^2(n))$ for all $0\leq i\leq \log_\varphi^*(n)$.

    It is clear that $\Time\left(n,\log^*_\varphi(n)\right)=1$.
    Now suppose $t$ is a layer-$i$ team where $i<\log^*_\varphi(n)$.  Each iteration in a recursive call $\expose(t,v)$ makes one call to $\cut$ and one call to $\link$. Both of these subroutine calls are made on teams in the next layer down, which by the inductive hypothesis takes $O(s_{i+2}(n)s^2_{i+3}(n))$. The iteration also recursively calls $\expose$ on a team in the next layer down. By the above argument this takes $s_{i+2}(n)\Time(n,i+1)$. Lastly the iteration also performs a fixed number of other elementary operations.
    Therefore we obtain the following expression for $0\leq i<\log^*_\varphi(n)$:
    \[
        \Time(n,i) \leq c_1 s_{i+2}(n)s^2_{i+3}(n)+s_{i+2}(n)\Time(n,i+1) + c_2
    \]
    where $c_1,c_2>0$ are constants.
    For convenience we drop the parameter $n$ in the above expression to get
    \begin{equation}\label{eqn:Time expose}
        \Time(i) \leq c_1s_{i+2}s^2_{i+3}+s_{i+2}\Time(i+1) + c_2
    \end{equation}
    Applying telescoping on (\ref{eqn:Time expose}), we obtain
    \begin{align*}
        \Time(0)    \leq    \ & c_1s_2s_3^2+c_1s_2s_3s^2_4+\cdots+c_1s_2\cdots s_{\log^*_\varphi(n)}s_{\log^*_\varphi(n)+1}s^2_{\log^*_\varphi(n)+2} & \\
                            \ & + c_2+c_2s_2+\cdots+c_2s_2\ldots s_{\log_\varphi^*(n)} &\\
                    \leq    \ & c_1\sum_{i=1}^{\log^*_\varphi(n)}\left(s_{i+2}\prod_{j=2}^{i+2}s_j\right) + c_2 \sum_{i=2}^{\log^*_\varphi(n)}\prod_{j=2}^i s_j\\
                    \leq    \ & c_1\sum_{i=1}^{\log^*_\varphi(n)}s_2s_3^2s_{i+2}+c_2\log^*_\varphi(n)s_2s^2_3 & \text{(by Lemma~\ref{lem:log product})}\\
                    \leq    \ & c_1\log^*_\varphi(n)s_2s_3^3+c_2\log^*_\varphi(n)s_2s^3_3  &
    \end{align*}
    Hence the running time of a single iteration in $\expose(\ell,u)$ is $O(\log^*_\varphi(n)s_2(n)s_3^3(n))$. By Lemma~\ref{lem:iterated log}, $\log^*_\varphi(n)$ is $O(s_3(n))$ and thus $\Time(n,0)$ is $O(s_2(n)s_3^4(n))$, which by (\ref{eqn:slog}), is $O(s_2^2(n))$.

\medskip

    \item[(b)] This statement follows directly from (a) and the fact that the maximum number of iterations performed by the $\expose(\ell,u)$ operation is $s_1(n)$.


\medskip

\item[(c)] For the $\link(\ell,\ell')$ operation we use the following inductive hypothesis: \ Any calls to $\cut$ and $\expose$ on teams at layer-1 of the LTT $\ell$ takes time $cs_2(n)s_3^2(n)$ for some constant $c>0$.

    Let $T$ and $T'$ be the top layer trees of $\ell$ and $\ell'$ respectively and $d(\ell,\ell')$ be the height difference between $T$ and $T'$.
    Recall that the $\link(\ell,\ell')$ operation finds a node $u$ on the rightmost path of $T$ such that $T(u)$ and $T'$ have the same height and links $T(u)$ and $T'$ to a new node below this node. Hence the $\expose(\ell,v)$ operation in $\link(\ell,\ell')$ consists of $d(\ell,\ell')$ iterations.
    By (a), this call to $\expose(\ell,v)$ takes time $c_1\cdot d(\ell,\ell')\cdot s_2^2(n)$, where $c_1$ is a constant.

    The $\link(\ell,\ell')$ operation also makes a call to $\rotateleft(\ell,y)$ which consists of three calls to $\cut$ and one call to $\expose$ on teams at a lower layer. By the inductive hypothesis, these subroutine calls to takes time $c_2s_2(n)s_3^2(n)$ for some constant $c_2>c$. The $\link(\ell,\ell')$ operation also performs $O(d(\ell,\ell'))$ many other elementary operations. Therefore the running time of $\link(\ell,\ell')$ is at most
    \[
        c_1d(\ell,\ell')s_2^2(n)+c_2s_2(n)s_3^2(n)+c_3d(\ell,\ell').
    \]
    Note that we may pick $c$ to be bigger than $c_1+c_3$ and therefore the above expression is at most
    \[
        (c_1+c_3)d(\ell,\ell')s_2^2(n)+c_2s_2(n)s_3^2(n)
    \]
    which is at most $c\cdot d(\ell,\ell')\cdot s_2^2(n)$ when $n$ is sufficiently large.
    Therefore the running time for $\link(\ell,\ell')$ is $O(d(\ell,\ell')s_2^2(n))$.

\medskip

\item[(d)] Let $T$ be the top-layer tree of $\ell$. The cut operation first makes  a call to $\changeval(T,u,-\infty)$, which by (b) takes time $O\left(s_1(n)s^2_2(n)\right)$.
It then walks the path from $u$ to the root. Let $P=\{u_0,u_1,\ldots,u_m\}$ be the path in $T$ from $u_0=u$ to the root of $T$ where $u_{i+1}=p(u_i)$ for all $0\leq i<m$.  It is clear that $m\leq s_1(n)$ and thus the traversal itself takes time $O(s_1(n))$.

By Alg.~\ref{alg:multicut}, the $\cut(\ell,u)$ operation separates $T$ into a collection of trees
\[
    \widehat{T}_1, \widehat{T}_2,\ldots,\widehat{T}_k
\]
where each $\widehat{T}_i$ is either the left or the right subtree of $u_i$. As $T$ is balanced, one could easily prove by induction on $i$ that
\[
    h\left(\widehat{T}_i\right) \leq 2i-1.
\]
The $\cut(\ell,u)$ operation then iteratively joins the trees $\widehat{T}_1, \widehat{T}_2,\ldots,\widehat{T}_k$ to form two trees $T_1$ and $T_2$. Let $n_i$ be the number of leaves in the tree $\widehat{T}_i$. By (c) the total running time of the sequence of $\link$ operations performed is at most
\begin{align*}
     & \ 2\sum_{i\geq 1}^{m-1} \left( h\left(\widehat{T}_{i+1}\right) - h\left(\widehat{T}_i\right)\right)\cdot s^2_2\left(n_{i+1}\right) \\
\leq & \ 2  \sum_{i\geq 1}^{m-1} \left( h\left(\widehat{T}_{i+1}\right) - h\left(\widehat{T}_i\right)\right)\cdot s^2_2(n)\\
\leq & \ 2\left(h\left(\widehat{T}_m\right)-h\left(\widehat{T}_1\right)\right) \cdot s^2_2(n)\\
\leq & \ 2s_1(n)s^2_2(n).
\end{align*}
Therefore the overall running time of the $\cut(\ell,u)$ operation is $O\left(s_1(n)s^2_2(n)\right)$. \qed
\end{enumerate}
\end{proof}

\begin{theorem}
There is an algorithm that solves the dynamic partial sorting problem which performs the $\psort(\ell,k)$ operation in time $O(\log^*_\varphi(n)k\log k)$, and performs the $\link(\ell,\ell')$, $\cut(\ell,x)$ and $\changeval(\ell,x,x')$ operations in time $O\left(\log (n)\cdot\log^2(\log(n))\right)$, where $n$ is the size of the list $\ell$.
\end{theorem}
\begin{proof} The correctness of the $\psort(\ell,k)$ operation follows from Lemma~\ref{lem:next0-correct}. For correctness of the update operation, assume that (J1)--(J3) hold for every node in the LTT data structure.
 Suppose we perform the $\changeval(\ell,u,x')$ operation. Since $u$ is a leaf in $\ell$, by Lemma~\ref{lem:changeval}, (J1)--(J3) still hold for every node in the LTT. Suppose we perform the $\link(\ell,\ell')$ operation. The $\expose(\ell,v)$ operation in Line~\ref{line:link-expose} in Alg.~\ref{alg:multilink} preserves (J1)--(J3) for every node. If the operation performs $\rotateleft(\ell,y)$ in Line~\ref{line:link-rotate}, then by Lemma~\ref{lem:multilink} (J1)--(J3) also hold for every node and thus $\link(\ell,\ell')$ is correct. Lastly, suppose we perform the $\cut(\ell,u)$ operation. Then (J1)--(J3) still hold by the correctness of $\changeval$ and $\join$.

 The complexity of the $\psort(\ell,k)$ operations follows directly from Lemma~\ref{lem:psort}. The complexity of the update operations follows from Lemma~\ref{lem:update} and Lemma~\ref{lem:number of layers}.\qed
\end{proof}

\section{Conclusion and Future Work}\label{sec:conc}
This paper presents data structures for solving the dynamic partial sorting problem.
%
We propose here two possible directions of optimizing the layered tournament trees: on query size and on intervals. In both cases, we seek to perform optimizations by determining an optimal query size or interval, and then creating a data structure which performs this query optimally.
This is similar in principle to optimized BSTs as presented in \cite{Cormen02}.

We can perform these optimizations either statically or dynamically. In the case of optimizing for query size, in the static case, we have a table of queries and the probability that a query will have that length (similarly to the optimal BST). We then determine an expected query length, and make a structure to perform queries of that length optimally. In the dynamic case, the structure keeps track of query probabilities, and dynamically rebuilds itself when the expected query length changes. When optimizing for intervals, one would take a similar approach, except to optimize access to a particular interval or set of intervals that are frequently queried.

As the layered tournament tree structure is designed for very large data sets, other optimizations to consider for the structure are parallelism, external memory use optimization, and persistence (as described in \cite{Haim02}). In particular, the first two of these are suitable for extremely large data sets, and require different analysis of the structure, and likely a different implementation as well.

\end{document}